\def\shortonly#1{#1}
\def\longonly#1{}

\documentclass[a4paper]{scrartcl}
\usepackage[pdfa]{hyperref}
\usepackage{booktabs}
\usepackage{authblk}
\usepackage{bibunits}
\usepackage{graphicx}
\usepackage{amsfonts}
\usepackage{amsthm}
\usepackage{multirow}
\usepackage{mathtools}

\newtheorem{theorem}{Theorem}
\newtheorem{definition}[theorem]{Definition}
\newtheorem{proposition}[theorem]{Proposition}

\DeclareOldFontCommand{\it}{\normalfont\itshape}{\mathit}
\DeclareOldFontCommand{\bf}{\normalfont\bfseries}{\mathbf}
\catcode`\@=11 
\DeclareOldFontCommand{\sc}{\normalfont\scshape}{\@nomath\sc}

\defaultbibliographystyle{plainurl}
\defaultbibliography{../../master}

\def\acosh{\mathrm{acosh}}

\def\map{\mu}
\def\dist{\delta}
\def\bbH{\mathbb{H}}
\def\ra{\rightarrow}
\def\tlimit{\gamma}

\def\gshort#1#2#3{G_{#1#2#3}}

\def\VofH{V}
\def\EofH{E}

\def\ghoz{\gshort 710}
\def\ghoo{\gshort 711}
\def\gooz{\gshort 810}

\newtheorem{thm}{Theorem}

\def\tally{\textsc{Pairs}}
\def\edgetally{\textsc{Edges}}

\title{Discrete Hyperbolic Random Graph Model}

\author{Dorota Celińska-Kopczyńska}
\author{Eryk Kopczyński}
\affil{Institute of Informatics, University of Warsaw, \texttt{\{\href{mailto:erykk@mimuw.edu.pl}{erykk},\href{mailto:dot@mimuw.edu.pl}{dot}\}@mimuw.edu.pl}}

\begin{document}

\maketitle

\begin{abstract}
The hyperbolic random graph model (HRG) has proven useful in the analysis of scale-free networks, which are ubiquitous in many fields, from social network analysis to biology. However, working with this model is algorithmically and conceptually challenging because of the nature of the distances in the hyperbolic plane.
In this paper, we propose a discrete variant of the HRG model (DHRG) where nodes are mapped to the vertices of a triangulation; our algorithms allow us to work with this model in a simple yet efficient way. We present experimental results conducted on networks, both real-world and simulated, to evaluate the practical benefits of DHRG in comparison to the HRG model.
\end{abstract}

\begin{bibunit}
Hyperbolic geometry has been discovered by 19th century mathematicians wondering about
the nature of parallel lines. One of the properties of this geometry is that the amount
of an area in the distance $d$ from a given point is exponential in $d$; intuitively, 
the metric structure of the hyperbolic plane is similar to that of an infinite binary
tree, except that each vertex additionally connects to two adjacent vertices on the same
level. 

Recently, hyperbolic geometry has proven useful in modeling hierarchical data
\cite{lampingrao,munzner}.
In particular, it has found application in the analysis of scale-free networks,
which are ubiquitous in many fields, from network analysis to biology \cite{papa}. 
In the hyperbolic random graph model (HRG), we place the nodes randomly in a hyperbolic disk;
nodes that are in a close neighbourhood are more likely to be connected.
The properties of a HRG,
such as its power-law degree distribution or high clustering coefficient, are similar to those of
real-world scale-free networks \cite{gugelman}. Due to high clustering coefficients, HRG is more accurate than  
earlier models such as Preferential Attachment \cite{prefatt} in modeling real-world networks.

Perhaps the two most important algorithmic problems related to HRGs are {\it sampling}
(generate a HRG) and {\it MLE embedding}: given a real-world network $H=(V,E)$, map the vertices of
$H$ to the hyperbolic plane in such a way that the edges are predicted as accurately as possible.
The quality of this prediction is measured with \emph{log-likelihood},
computed with the formula
$\sum_{v,w} \log p(v,w)$, 
where $p(v,w)$ is the probability that the model correctly predicts the existence or not of an edge $(v,w)$.
Those problems are non-trivial: even simply computing the log-likelihood,
using a naive algorithm would require time $O(|V|^2)$.
The original paper \cite{papa} used an $O(|V|^3)$ algorithm for embedding. Efficient algorithms have been found
for generating HRGs and the closely related Geometric Inhomogeneous Random Graphs in expected time $O(|V|)$ \cite{gengraph,efficgen,penschuck,loozgen,loozphd,funkegen} and for MLE embedding real-world scale-free networks
into the hyperbolic plane  in time ${\tilde O}(|V|)$ \cite{tobias}, which was a major improvement over 
previous algorithms \cite{hypermap,vonlooz}. 
Embedding has practical applications in link prediction \cite{hyperblink} and routing \cite{bogu_internet,tobias_alenex}.

This paper introduces and experimentally studies the discrete version of the HRG model (DHRG).
In the DHRG model, we use a tessellation rather than the hyperbolic plane.
Instead of the hyperbolic
distance between points, we use the number of steps between two tiles in the tessellation.
Our approach has the following advantages:

\begin{itemize}
\item Avoiding numerical issues. DHRG is not based on the tuple of coordinates, which makes it immune to serious precision errors resulting from the exponential expansion. This way we solve a fundamental issue for hyperbolic embeddings \cite{tobias_alenex,reptradeoff}.
\item Algorithmic simplicity. In DHRG we find efficient algorithms for
sampling, computing the log-likelihood, and improving an embedding 
by generalizing similar algorithms for trees. Working with DHRG does not require a good understanding of hyperbolic geometry,
combating a major drawback of previous approaches.
\end{itemize}

The first potential objection to our approach is that discrete distances are inaccurate.
This inaccuracy comes from two sources: different geometry (distances in Euclidean square grid
correspond to the taxicab metric, which is significantly different from the usual Euclidean
metric) and discreteness. It is challenging to rigorously study the theoretical effects of discretization
on the properties of our model.
However, according to our experiments, these issues do not turn out to be threatening --- the
HRG embeddings are large enough to render discreteness insignificant, and discrete distances
are a good approximation of the actual hyperbolic distances (better than in the case of
Euclidean geometry).

Our experiments on artificial networks show that using DHRG improves the success rate of greedy routing in 77\% of the cases (depending on the parameters).
A DHRG embedding can be efficiently improved by moving the vertices
so that the log-likelihood becomes better. 
Our procedure yields about 10\% improvement of log-likelihood on state-of-the-art HRG embeddings on
real graphs. This result is supported by our extensive simulations on artificial graphs.

\section{Prerequisities}

In this section, we briefly introduce hyperbolic geometry and the HRG model.
A more extensive introduction to hyperbolic geometry can be found, e.g.,~in \cite{cannon}.

\begin{figure}
\begin{center}
\rotatebox{90}{\includegraphics[width=.3\textwidth]{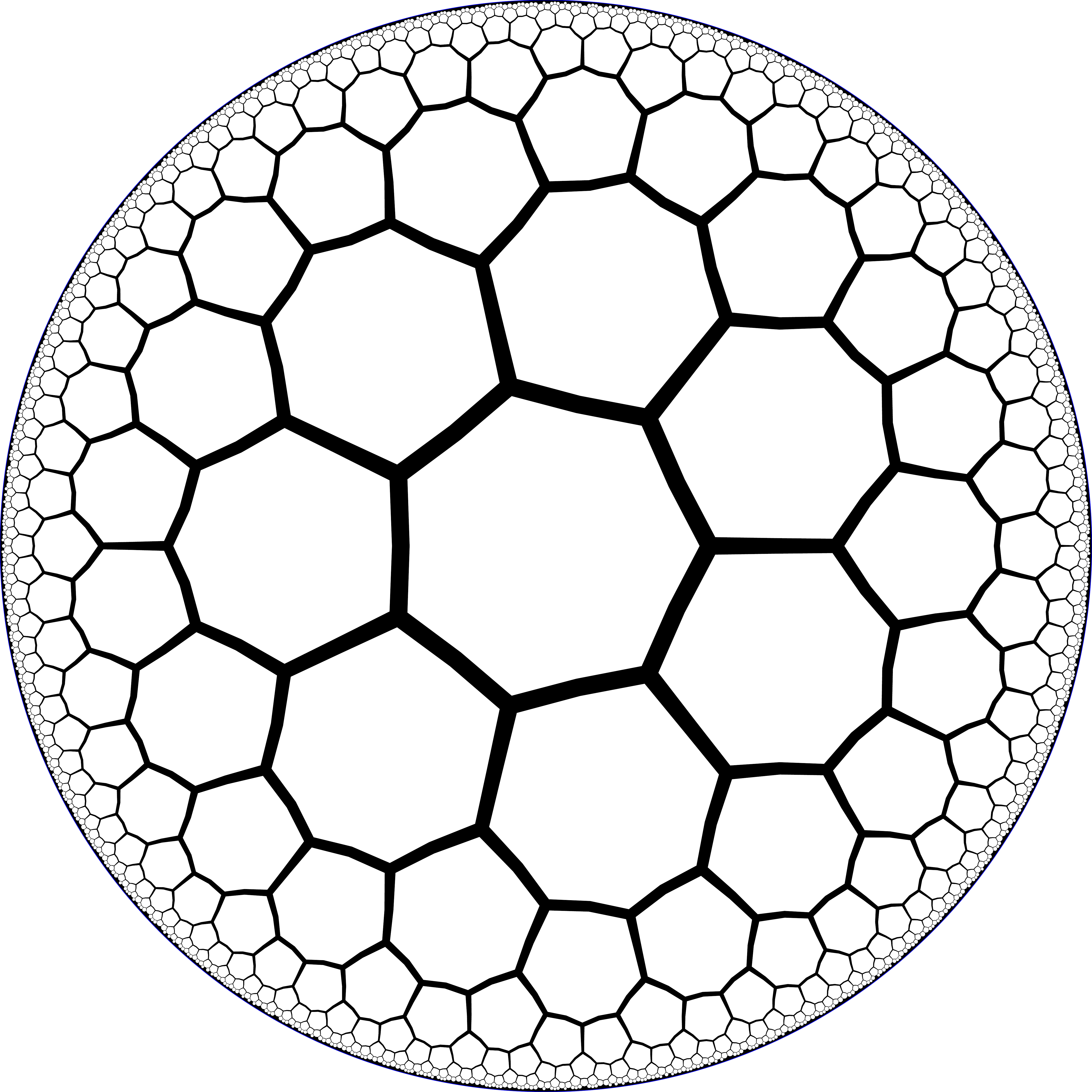}}
\includegraphics[width=.3\textwidth]{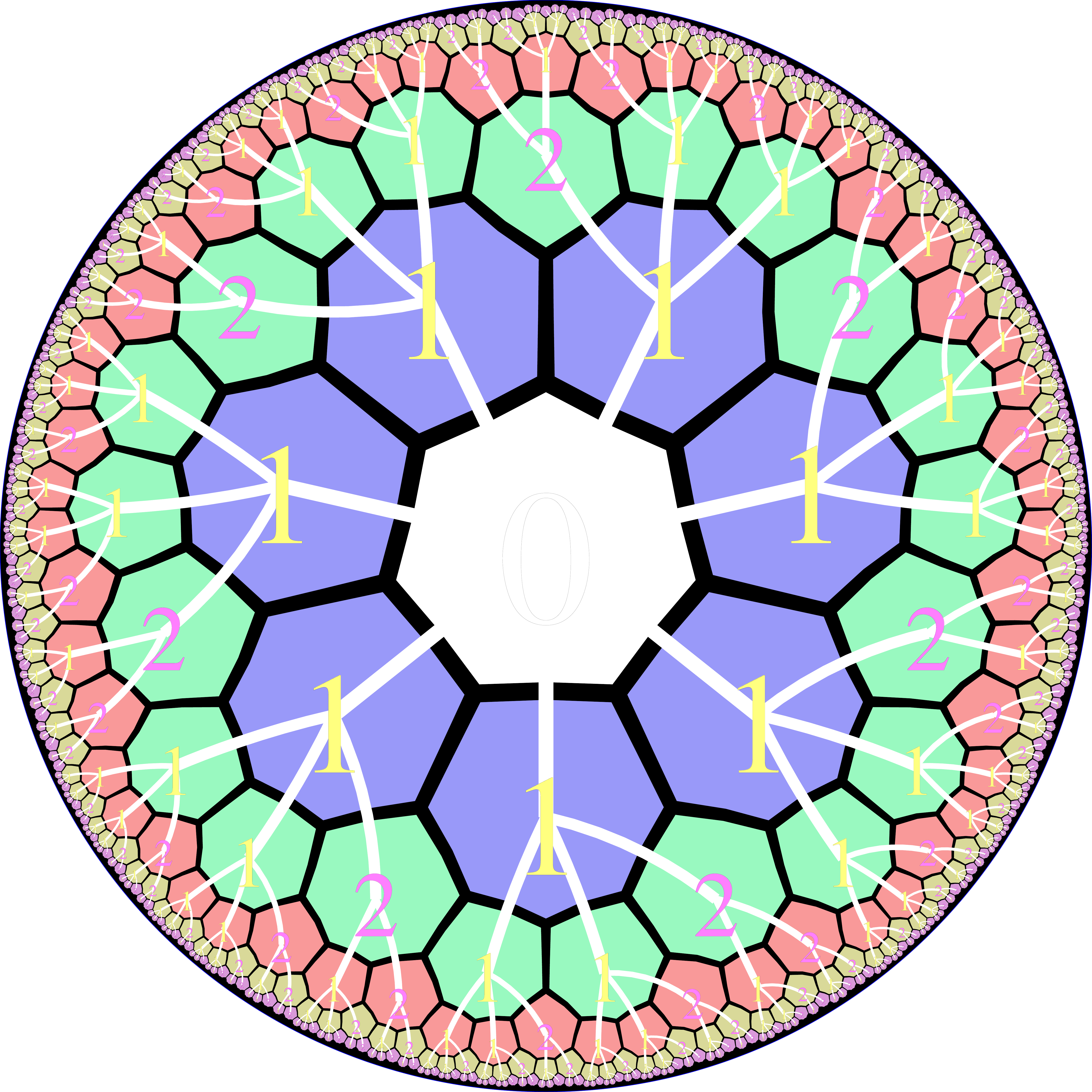}
\includegraphics[width=.3\textwidth]{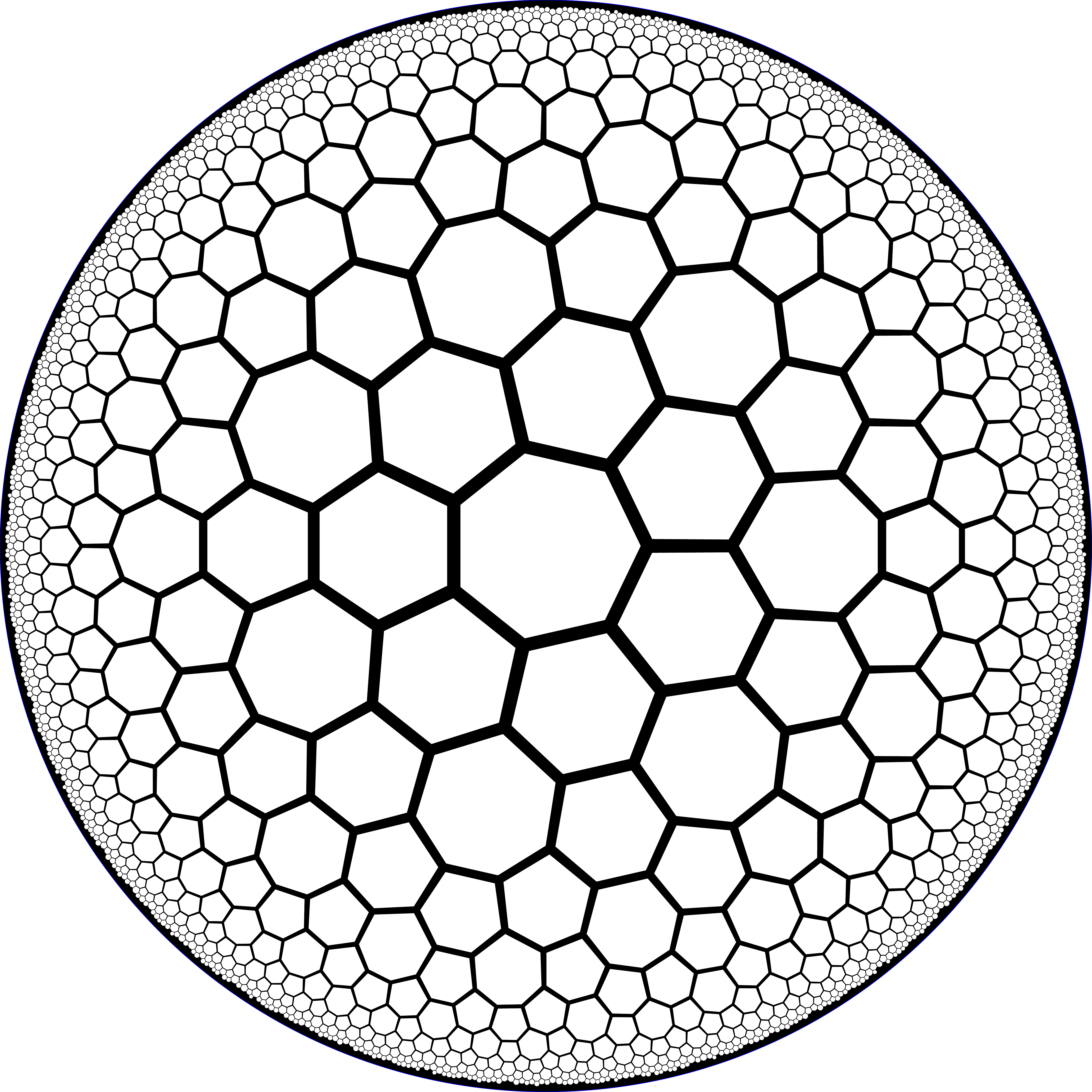}
\end{center}
\caption{\label{figtile}\label{growgraph}
(a) Order-3 heptagonal tiling ($\ghoz$). (b) Its growth. (c) Bitruncated heptagonal tiling ($\ghoo$).}
\end{figure}

Figure \ref{figtile} shows the order-3 heptagonal tessellation of the hyperbolic plane in the Poincar\'e model. In the hyperbolic metric, 
all the triangles, heptagons, and hexagons on each picture are actually of the same size.
The points on the boundary of the disk are infinitely far from the center. The area
of a hyperbolic circle of radius $r$ is exponential in $r$.

The hyperbolic plane (in the Minkowski hyperboloid model) is
\(\bbH^2 = \{(x,y,z): z>0, z^2-x^2-y^2=1\}.\)
The distance between two points $a=(x,y,z)$ and $a'=(x',y',z')$ is \(\delta(a,a') = \acosh(zz'-xx'-yy').\)
Most introductions to hyperbolic geometry
use the Poincar\'e disk model; however, the Minkowski hyperboloid model is very useful
in computational hyperbolic geometry, because the essential operations are simple generalizations
of their Euclidean or spherical counterparts. In particular, rotation of a point $v$ by angle $\alpha$ is given
by \[\left(\begin{array}{ccc}\cos\alpha & \sin\alpha & 0 \\ -\sin\alpha & \cos\alpha & 0 \\ 0 & 0 & 1\end{array} \right)v,\]
while a translation by $x$ units along the $X$ axis is given by 
\[\left(\begin{array}{ccc}\cosh x & 0 & \sinh x \\ 0 & 1 & 0 \\ \sinh x & 0 & \cosh x \end{array} \right)v.\]
We can easily map the Minkowski hyperboloid model
to the Poincar\'e disk model (e.g., for visualization purposes) using stereographic projection:
\[(x',y') = (x/(z+1), y/(z+1)).\]
In the hyperbolic polar coordinate system, every point is represented by two coordinates $(r,\phi)$,
where \[P(r,\phi) = (\cos(\phi)\sinh(r), \sin(\phi)\sinh(r), \cosh(r)).\] Here, $r$ is distance from the
central point $(0,0,1)$, and $\phi$ is the angular coordinate.

\begin{definition}
\label{def_hrg}
The {\bf hyperbolic random graph model}
has four parameters: $n$ (number of vertices), $R$ (radius), $T$ (temperature), and $\alpha$ (dispersion parameter). Each vertex $v \in \VofH = \{1,\ldots,n\}$
is independently randomly assigned a point $\map(v) = P(r_v, \phi_v)$, where
the distribution of $\phi_v$ is uniform in $[0,2\pi]$, and
the density of the distribution of $r_v \in [0,R]$ is given by 
$f(r) = { {\alpha \sinh(\alpha r)} \over {\cosh(\alpha R)-1}}$.
Then, for each pair of vertices
$v,w\in \VofH$, they are independently connected with probability $p(\dist(\map(v),\map(w)))$,
where $\dist(x,y)$ is the distance between $x,y \in \bbH^2$, and $p(d) = {1 \over 1+e^{(d-R)/2T}}$.
\end{definition}

The parameter $\alpha$ controls the power-law exponent $\beta = 2\alpha+1$ \cite{gugelman}. The parameter $T$, typically
chosen to be in $[0,1]$, regulates the
importance of underlying geometry, and thus the clustering coefficient: with $T$ very close to 0, an edge exists iff $\dist(\map(v), \map(w))<R$, while 
with larger values of $T$ missing short edges and existing long edges are possible. In \cite{gugelman} and \cite{tobias}, $R$ equals $2\log n + C$,
where $C$ is a parameter adjusting the average degree of the resulting graph.

An {\bf MLE embedder} is an algorithm which, given a network $H=(V,E)$, finds a good embedding of $H$ in the hyperbolic
plane, i.e., parameters $R$, $T$, and $\alpha$, and a mapping $\map:V \ra \bbH^2$. The quality of the embedding
is measured with \emph{log-likelihood}, computed with the formula
\[\log L(\map) = \sum_{v<w \in \VofH} \log p_{\{v,w\}\in E}(\dist(\map(v), \map(w))),\]
where $p_\phi(d)=p(d)$ if $\phi$ is true and $1-p(d)$ if $\phi$ is false.
While not a goal by itself, we can expect that a better embedding (in terms of log-likelihood) will perform better in the applications, such as link prediction, greedy routing, visualization, etc. 

Now, we explain the structure of a hyperbolic tessellation, on the example of the order-3 heptagonal tiling
($\ghoz$ from Figure \ref{growgraph}b). Let $\delta(t_1,t_2)$ be the distance between two tiles ($\delta(t_1,t_2)=1$ iff $t_1$ and $t_2$ are adjacent).
Let $\delta_0(t)=\delta(t,t_0)$ be the distance of tile $t$ from the center tile $t_0$.
We denote the set of tiles $t$ such that $\delta_0(t)=d$ with $R_d$; for $d>0$ it is a cycle (in Figure \ref{growgraph}b sets $R_d$
are marked with colored rings).
Except for the central tile, every tile has one or two adjacent tiles in the previous layer
(called left and right parent, $p_L$ and $p_R$), two adjacent tiles in the same layer (left and right sibling, 
$s_L$ and $s_R$), and the remaining tiles in the next layer (children). By connecting every tile to its (right) parent,
we obtain an infinite tree structure. The numbers 0, 1, 2 denote the \emph{type} of the tile, which is
the number of parents. We can implement the function $\sc{Adj}(t,i)$ returning the pointer to the $i$-th neighbor to tile $t$,
clockwise starting from the parent in amortized time $O(1)$ by using a lazily generated representation
of the tessellation, where each tile is represented by a node holding pointers to the parent node, children nodes
(if already generated), and siblings (if already known).
Since every tile has at least two (non-rightmost) children, this structure grows exponentially. To gain intuition about hyperbolic tessellations, we recommend playing HyperRogue \cite{hyperrogue} as its gameplay focuses on the crucial concepts of exponential growth and distances in the tessellation graph.

In a Euclidean tessellation, the distance $x$ between the center of two adjacent tiles can be arbitrary; the
same tessellation can be as coarse or fine as needed.
This is not the case for hyperbolic tessellations. In the case of $\ghoz$, $x$ must be the edge length of a triangle with angles
$\pi/7$, $\pi/7$, and $2\pi/3$, which we can find using the hyperbolic cosine rule. If the central tile is
at Minkowski hyperboloid coordinates $(0,0,1)$, the coordinates of every other tile $t$ can be found by composing
translations (by $x$ units) and rotations (by multiples of $2\pi/7$).



\section{Our contribution}\label{sec:scale-free}

Here we introduce 
the discrete hyperbolic random graph model (DHRG), which is the discrete version of the HRG model (Definition \ref{def_hrg}). We map vertices $v \in \VofH$ not to points in the continuous
hyperbolic plane but the tiles of our tessellation,
i.e., $\map:\VofH \ra D_R$, where $D_R$ is the set of all tiles in distance at most $R$.

\begin{definition} A {\bf discrete hyperbolic random graph (DHRG)} with parameters
$n,$ $R,$ $T,$ and $\alpha$ is a random graph $H=(\VofH,\EofH)$ constructed as follows:
\begin{itemize}
\item The set of vertices is $\VofH = \{1, \ldots, n\}$,
\item Every vertex $v \in \VofH$ is independently randomly assigned a tile $\map(v) \in D_R$,
in such a way that the probability that $\map(v)=w \in R_d$ is proportional to $\frac{e^{d\alpha}}{|R_d|}$;
\item Every pair of vertices $v_1, v_2 \in \VofH$ are independently connected with an edge
with probability $p(\dist(\map(v_1),\map(v_2)))$, where $p(d) = {1 \over 1+e^{(d-R)/2T}}$.
\end{itemize}
\end{definition}

Note that the definition permits $\map(v_1)=\map(v_2)$ for two different vertices $v_1,v_2 \in \VofH$. This is not a problem; such vertices $v_1$ and $v_2$ are not necessarily connected,
nor do they need to have equal sets of neighbors. This may happen when two vertices are too similar to 
be differentiated by our model.

Determining the relation between discrete and continuous distances
theoretically is challenging. We find this relation experimentally. We compute the hyperbolic distance $R_d$ between $h_0 = (0,0,1)$
and $T_d$, where $T_d$ is a randomly chosen tile such that $\dist_0(T_d) = d$. We get $R_d = c_1 d  + c_2 + X_d$,
where $c_1 \approx 0.9696687$, $c_2 \approx 0.0863634$,
and $X_d$ is a random variable with a bell-shaped distribution, expected value $EX_d = o(1)$ and variance $\mbox{Var}\ X_d = \Theta(d)$.
\longonly{Table \ref{distversus} contains} 
\shortonly{(See Appendix \ref{appendix:versus} for}
the example detailed results for $\ghoz$.) 
\longonly{The asymptotic values are obtained in the following way: the difference $ER_d-ER_{d-1}$ converges very quickly to 0.9696687, and the difference
$ER_d-0.9696687d$ converges very quickly to 0.0863634.}
Note that the error is much better than for Euclidean hexagonal grid,
where a similar formula holds, but with $\mbox{Var}\ X_d = \Theta(d^2)$.
This is because in the Euclidean plane, the stretch factor depends on the angle between the line $[h_0, T_d]$
and the grid, which remains constant along the whole line; on the other hand, in the hyperbolic plane, this angle is not constant, and its values
in sufficiently distant fragments of that line are almost independent. 
\longonly{The number $n_d$ in Table \ref{distversus} is}
\shortonly{Let $n_d$ be}
the number of tiles in distance $d$;
they are every second Fibonacci number multiplied by 7, and thus $n_d = \Theta(\gamma^d)$ for $\gamma = \frac{3+\sqrt{5}}{2}$. 
The radius of a hyperbolic disk which has the same area as the union of all tiles in distance at most $d$ is $\log(\gamma)d + O(1)$; 
our coefficient $c_1$ is close to $\log(\gamma) \approx 0.9624237$, but slightly larger. 

Let $j$ be a function which maps every tile of our tessellation to the coordinate of its center. DHRG mappings can be converted to HRG by composing
$\map$ with $j$, and the other conversion can be done by finding the tile containing $\map(v)$ for each $v\in \VofH$. 
Since our experimental results show that the discrete distances are very good approximations of the continuous distances (up to the multiplicative constant $c_1$
and additive constant $c_2$), we expect the desired properties of HRGs, such as the high clustering coefficient and the power-law degree distribution, to still be true in DHRGs.
The parameters $\alpha$, $R$ and $T$ of the DHRG model will be obtained from the HRG parameters by dividing them by $c_1$.


\section{Algorithms for DHRG}\label{sec:dhrgalgo}

In this section, we present our algorithms for working with the DHRG model.
Our algorithms will be efficient under the assumption $R=O(\log n)$ and $m=o(n^2/R)$.

\begin{proposition}\label{candodist}
There is a canonical shortest path between
every pair of tiles $(t_1, t_2)$. If $t_2$ is to the right from $t_1$, this canonical 
shortest path consists of: a number of right parent edges; at most one right sibling edge;
and a number of non-leftmost child edges.
If $t_1$ is to the right from $t_2$, the canonical path is defined symmetrically. For any pair
of tiles $(t_1,t_2)$, the distance $\delta(t_1,t_2)$ can be found in time $O(\delta(t_1,t_2))$.
\end{proposition}

The algorithm for finding $\delta(t_1,t_2)$ works as follows: for every $d$ starting from
$\min(\delta_0(t_1), \delta_0(t_2))$ and going downwards, we find the leftmost and rightmost ancestors of $t_1$ and $t_2$
in $R_d$. 
If one of the ancestors of $t_1$ matches one of the ancestors of $t_2$, we return
$\delta_0(t_1)+\delta_0(t_2)-2d$; if they do not match but are adjacent, we return
$\delta_0(t_1)+\delta_0(t_2)-2d+1$.

In our application, we will need to efficiently find the distance between a tile $t$ and all tiles $u \in A$.
We will do it using the following data structure:

\begin{definition}
A \emph{distance tally counter} is a structure with the following operations:

\begin{itemize}
\item {\sc Init}, which initializes the multiset of tiles $A$ to empty.
\item {\sc Add}($u$,$x$), which adds the tile $u$ to the multiset $A$ with multiplicity $x$ (which can be negative).
\item {\sc Count}($t$), which, for tile $t$, returns an array $T$ such that $T[d]$ is the number of elements of $A$ in distance $d$ from $t$.
\end{itemize}
\end{definition}

\begin{thm}\label{tallythm}
There is an implementation of distance tally counter where all the operations are executed in $O(R^2)$,
where $R$ is the maximum distance from the central tile.
\end{thm}

\begin{proof} (sketch)
A {\bf segment} is a pair of tiles of form $[p_L^k(t), p_R^k(t)]$ for some tile $t$ and $k \geq 0$.
The notation $p_R^k$ here denotes the $k$-th iteration, i.e., the rightmost $k$-th ancestor in this case.
For a segment $s=[v_L,v_R]$, let $p(s)=[p_L(v_L), p_R(v_R)]$ be the parent segment.
In the case of $\ghoz$, either $p_L^k(t)=p_R^k(t)$, or they are neighbors. 
The algorithm from Proposition \ref{candodist} can be seen as follows:
we start with two segments $[t_1,t_1]$ and $[t_2,t_2]$, and
then apply the parent segment operation to each of them until we obtain segments which
are close. To construct an efficient distance tally counter, we need to tally the ancestor segments for
every $u \in A$.

For every segment $s$, we keep an array $a_s$, where $a_s[d]$ is the number of $u \in A$ such that
$p^d[u,u] = s$. The operation {\sc Add} updates these arrays in time $O(R^2)$. The operation
{\sc Count}($t$) constructs $s_d=p^d[t,t]$ for $d=0,\ldots,\delta_0(t)$, and uses the information 
in segments intersecting or adjacent to $s_d$ to count the number of elements of $u$ for which
the distance algorithm would return every possible distance. We need to make sure that we do not
count the same $u \in A$ twice (for different values of $d$).
However, this can be done by 
temporarily subtracting from $a_s[d]$ entries which correspond to $u$'s which have been already counted;
see Appendix \ref{appendix:proofs} for details.
\end{proof}

Other tessellations than the order 3 heptagonal tessellation are possible.
The order-3 octagonal grid, $\gooz$, is coarser.
Finer tessellations can be obtained by applying the Goldberg-Coxeter construction,
such as $\ghoo$ from Figure \ref{growgraph};
their growth is less extreme than for $\ghoz$, and thus the distance between segment ends, as well as the
number of sibling edges in the canonical path, may be greater than 1\longonly{; but they are still bounded}. 
\longonly{Let $D(G)$ be the bound; for tessellations useful in practice the value of $D(G)$ is small, we have $D(\ghoz)=1$ and $D(\ghoo)=3$.}
However, the algorithms generalize; see Appendix \ref{appendix:proofs}.


{\bf Computing the likelihood.} 
Computing the log-likelihood in the continuous model is difficult, because
we need to compute the sum over $O(n^2)$ pairs; a better algorithm
was crucial for efficient embedding of large real-world scale-free networks \cite{tobias}.
The algorithms above allow us to compute it quite easily and efficiently in the DHRG model.
To compute the log-likelihood of our embedding of a network $H$ with $n$ vertices and
$m$ edges, such that $\dist_0(v) \leq R$ for each $v \in \VofH$, we:

\begin{itemize}
\item for each $d$, compute $\tally[d]$, which is the number of pairs $(v,w)$ such that $\dist(v,w) = d$. The distance tally counter allows doing this in a straightforward way
({\sc Add}($\map(v)$,1) for each $v \in \VofH$ followed by {\sc Count}($\map(v)$) for each $v \in \VofH$), in time
$O(n R^2)$.

\item for each $d$, compute $\edgetally[d]$, which is the number of pairs $(v,w)$ connected by an edge such
that $\dist(v,w) = d$. This can be done in time $O(mR)$ simply by using the
distance algorithm for each of $m$ edges.
\end{itemize}

After computing these two values for each $d$, computing the log-likelihood is
straightforward. One of the advantages over \cite{tobias} is that we can then
easily compute the log-likelihood obtained from other values of $R$ and $T$,
or from a function $p(d)$ which is not necessarily logistic.

{\bf Improving the embedding.} A continuous embedding of good quality can be obtained by
first finding an approximate embedding and then improving it using a {\it
spring embedder} \cite{springembedder}. Imagine there are attractive forces
between connected pairs of vertices, and repulsive forces between unconnected pairs.
The embedding $m$ changes in time as the forces push the vertices towards locations
in such a way that the quality of the embedding is improved. Computationally, spring embedders are very expensive --- there are $\Theta(n^2)$
forces, and potentially many steps of simulation could be necessary.

On the contrary, our approach allows to improve DHRG embeddings easily.
We use a local search algorithm.
Suppose we have computed the log-likelihood and on
the way we have computed the vectors $\tally$ and $\edgetally$, as well as the distance
tally counter where every $\map(v)$ has been added. Let $v' \in \VofH$ be a vertex
of our embedding, and $w$ be a tile. Let $\map'$ be the new embedding given by
$\map'(v')=w$ and $\map'(v)=\map(v)$ for $v \neq v'$. Our auxiliary data lets us
compute the log-likelihood of $\map'$ in time $O(R^2 + R \deg(w))$. 

We try to improve the embedding in the following way: in each iteration,
for each vertex $v \in \VofH$, consider all neighbors of $\map(v)$, compute the log-likelihood for all
of them, and if for some $\map'$ we have $\log L(\map') > \log L(\map)$, 
replace $\map$ with $\map'$. Each iteration takes time $O(R^2n + Rm)$.

\section{Experimental setup}\label{sec:experiments}

The setup of the experiments is as follows. First, we map a network to the hyperbolic plane using the hyperbolic embedder based on the
algorithm from \cite{tobias} (for brevity, we will call it BFKL). This is the embedding stage. We start with a default parameters for the embedder ($R_0, T_0, \alpha_0$). This way we obtain the placement of the nodes in the hyperbolic plane, on which we estimate HRG predictive model with the same $R_0, T_0, \alpha_0$ as the embedder (prediction stage). We compute the log-likelihood ($L_1$). Usually, $L_1$ is even worse than the log-likelihood
of the na{\"\i}ve Erd\"os-R\'enyi-Gilbert model where each edge exists with probability $m/{n \choose 2}$. This is because the influence of the parameter $T$
on the quality of the embedding is small \cite{hypermap} and BFKL uses a small value of $T=0.1$ that does not necessarily correspond to the network. The prediction stage is irreplaceable here --- the log-likelihood is computed for the predictive model that takes the placement of the nodes as given. That is why, we should be still able to obtain a better log-likelihood during the prediction stage by estimating HRG models with different parameters ($R_1, T_1$); for brevity, the log-likelihoods obtained via such an optimization will be called ``the best log-likelihoods''. $L_2$ is the best log-likelihood obtained with the default embedding. As it proxies the best possible outcome for the default embedding, it will serve as the benchmark scenario in our experiment. 

Now we need to check if using DHRG improves the quality of the embedding. To this end, we convert our embedding into the DHRG model, by finding the nearest tile 
of our tessellation for each $v \in \VofH$. $L_3$ is the best discrete log-likelihood for a new embedding (computed with the logistic function). We call this phase \emph{discretization}.

Next, we try our \emph{local search} algorithm (20 iterations) and compute the best discrete log-likelihood after local search ($L_5$).%
\footnote{We do 20 iterations because further iterations tend to move less and less vertices, and thus their effect on the quality of embedding is minimal.}
Finally, we proceed to the \emph{de-discretization} phase, i.e., convert our mapping back to the HRG model and find the best log-likelihood $L_7$.

We denote the running times as $t_m$ (converting HRG to DHRG), $t_l$ (computing $\tally$ and $\edgetally$),
$t_e$ (local search). For comparison, we also include the time of computing the best continuous log-likelihood
$t_c$ using a parallelized $O(n^2)$ algorithm\footnote{We compute the distance for every pair of nodes. This way we create the array $d$, where $d[n]$ is the number of distances in the interval
$[n\varepsilon, (n+1)\varepsilon)$; we take $\varepsilon=10^{-4}$. This computation is parallelized. The array $d$ allows computing a very good approximation of the log-likelihood for any value of parameters $R$ and $T$ in time
$O(R/\varepsilon)$.)}, and the time of computing the log-likelihood by the BFKL algorithm ($t_b$).
All the times are measured on Intel(R) Core(TM) i7-9700K CPU @ 3.60GHz with 96 GB RAM.
Most computations use a single core, except the continuous log-likelihood values ($L_2$ and $L_7$)
which use 8 cores. 
Our implementation, using the RogueViz non-Euclidean geometry engine \cite{rogueviz2021}, 
as well as the results of our experiments can be found at 
\url{https://figshare.com/articles/software/Discrete_Hyperbolic_Random_Graph_Model_code_and_data_/16624369}.

\section{Experiments on real-world networks}

\begin{table*}
\centering
\resizebox{\linewidth}{!}{
\begin{tabular}{|l|rrrr|l|rrrr|rrrrrr|}
\hline
name & $n$ & $m$ & $R$ & $\alpha$ & grid & -$L_2$ & $\frac{L_3}{L_2}$ & $\frac{L_5}{L_3}$ & $\frac{L_7}{L_2}$ & MB & $t_m$ [s] & $t_l$ [s] & $t_e$ [s] & $t_c$ [s] & $t_b$ [s]\\
\hline
Fb & 4309 & 88234 & 12.57 & 0.755 & $G_{710}$ & 176131 &  1.04 &  0.93 &  0.97 & 40 & 0.196 &  0.03 & 10 & 0.35 & 0.048\\
Fb & 4309 & 88234 & 12.57 & 0.755 & $G_{810}$ & 176131 &  1.07 &  0.92 &  0.98 & 54 & 0.183 &  0.03 & 8 & 0.5 & 0.048\\
F09 & 74946 & 537972 & 20.90 & 0.855 & $G_{710}$ & 3954627 &  1.04 &  0.86 &  0.90 & 2010 & 5.432 &  1.16 & 222 & 131 & 0.896\\
F09 & 74946 & 537972 & 20.90 & 0.855 & $G_{810}$ & 3954627 &  1.06 &  0.84 &  0.90 & 1866 & 4.634 &  0.81 & 176 & 128 & 0.896\\
Sd & 77352 & 327431 & 26.00 & 0.610 & $G_{710}$ & 2091651 &  1.25 &  0.72 &  0.92 & 2659 & 5.326 &  1.05 & 201 & 130 & 0.292\\
Sd & 77352 & 327431 & 26.00 & 0.610 & $G_{810}$ & 2091651 &  1.27 &  0.71 &  0.92 & 2253 & 4.618 &  0.78 & 158 & 126 & 0.292\\
Am & 334863 & 925872 & 24.11 & 0.995 & $G_{710}$ & 6957174 &  1.04 &  0.86 &  0.91 & 5677 & 23.34 &  5.40 & 721 & 2690 & 1.444\\
Am & 334863 & 925872 & 24.11 & 0.995 & $G_{810}$ & 6957174 &  1.04 &  0.85 &  0.90 & 4868 & 19.76 &  3.92 & 576 & 2811 & 1.444\\
F11 & 405270 & 2345813 & 26.34 & 0.715 & $G_{710}$ & 20028756 &  1.22 &  0.76 &  0.93 & 9995 & 30.36 &  7.36 & 1349 & 3715 & 5.216\\
F11 & 405270 & 2345813 & 26.34 & 0.715 & $G_{810}$ & 20028756 &  1.22 &  0.76 &  0.93 & 8940 & 25.84 &  5.38 & 1113 & 3636 & 5.216\\
Go & 855804 & 4291354 & 26.06 & 0.865 & $G_{710}$ & 22762281 &  1.30 &  0.75 &  0.98 & 18226 & 64.75 & 16.05 & 2363 & 16618 & 3.560\\
Go & 855804 & 4291354 & 26.06 & 0.865 & $G_{810}$ & 22762281 &  1.32 &  0.75 &  0.99 & 15314 & 54.31 & 10.93 & 1823 & 15818 & 3.560\\
Pa & 3764118 & 16511741 & 28.74 & 0.995 & $G_{810}$ & --- & --- &  0.90 & --- & 66396 & 250.6 & 73.65 & 9335 & --- & 41.24\\
\hline
\end{tabular}} 
\caption{Experimental results on real-world networks.
Facebook (Fb), Slashdot (Sd), Amazon (Am), Google (Go), and Patents (Pa) networks from SNAP database; F09 and F11 are GitHub networks. MB is the amount of memory in megabytes, and time is in seconds.
\label{tab:realworld}
}
\end{table*}
  
Table~\ref{tab:realworld} contains detailed results of the experiments.
The networks we use come mostly from SNAP database \cite{snapnets}. To benchmark our algorithm on a large network,
we additionally study undirected social networks with power-law-like scale behavior,
representing the following relations that occured between 2009 and 2011 in GitHub \cite{euromed} (see Appendix \ref{appendix:github} for the details).
It makes sense to use finer tessellations for smaller graphs and coarser tessellations for larger ones.
In most cases, we conduct our experiments on two tessellations: $\ghoz$ (order-3 heptagonal) and the coarser $\gooz$ (order-3 octagonal).
For the Facebook graphs, we also try finer tessellations; see Appendix \ref{appendix:choosetes} for the details.
Finer tessellations give better log-likelihoods, but a too dense grid dramatically decreases the performance without giving significant benefits.

The parameters $n$, $m$, $R$ and $\alpha$ come from the BFKL embedder ($n$ is the number of vertices, $m$ is the number of edges). Discretization worsens the log-likelihoods; for smaller networks $L_3$ are usually slightly worse than $L_2$, but this is not surprising. First, our edge
predictor has lost some precision in the input because of our tesselation's discrete nature. Second, the original prediction was based on the hyperbolic distance $r$ while our prediction
is based on the tesselation distance $d$, and the ratio of $r$ and $d$ depends on the direction.
The whole procedure improves the log-likelihood by up to 10\%. For the patents network, computing $L_2$ and $L_7$ was not feasible; $L_3$ was 208618134.

The current version uses a significant amount of RAM. 
It should be possible to improve this by better memory management (currently vertices and segments which are
no longer used or just temporarily created are not freed), or possibly path compression.

When it comes to the running time, computing $\edgetally$ and $\tally$ takes negligible time ($t_l$) when compared to the network size. Even for as large network as Patents, those operations took slightly over a minute (Table~\ref{tab:realworld}); for most of the networks analyzed, they took a few seconds. Converting time $t_m$ increases with the size of the network, for most of the networks analyzed we need less than a minute. In comparison to computing log-likelihood as performed by BFKL, our solution is comparable in time; they are of the same order of magnitude. The longest time is needed for local search. Not surprisingly, the larger graph, the longer it takes to find improvements.
%
%
However, their spring embedder working in quadratic time is much slower
than our local search.\footnote{
We have also ran the BFKL spring embedder on the Facebook graph for $T=0.54336$ and seed 123456789,
reporting the log-likelihood of -131634, better than ours. However, 
this result appears incorrect; our implementation
reports the original log-likelihood of $L_1=-211454$, and can improve it to $L_7=-157026$.
Computing the log-likelihood incorrectly may negatively impact the quality of BFKL embedding.
}



\section{Experiments on simulated graphs}

\subsection{Log-likelihood} \label{exp:ll}
Our experiments on real-world graphs showed that discretization, local search, and de-discretization improves the quality of the embedding in terms of log-likelihood.
In this section we conduct simulations to see whether we can generalize those observations.

We use the generator included with \cite{tobias} to generate HRGs with the following parameters: varying $n$,
$\alpha=0.75$, $T=0.1$, $R=2\log(n)-1$. These are the default values of parameters used by this generator.
For every value of $n$ considered, we generate 1000 graphs.
For each of the generated graphs $H$, we embed $H$ using BFKL, compute $L_2$, convert to DHRG on tessellation $\ghoz$,
improve the embedding (up to 20 iterations), convert back to HRG, and compute $L_7$. We also 
compute $L_g$, which is the log-likelihood of the originally generated embedding (groundtruth).

\begin{figure}
\begin{center}
\includegraphics[width=\linewidth]{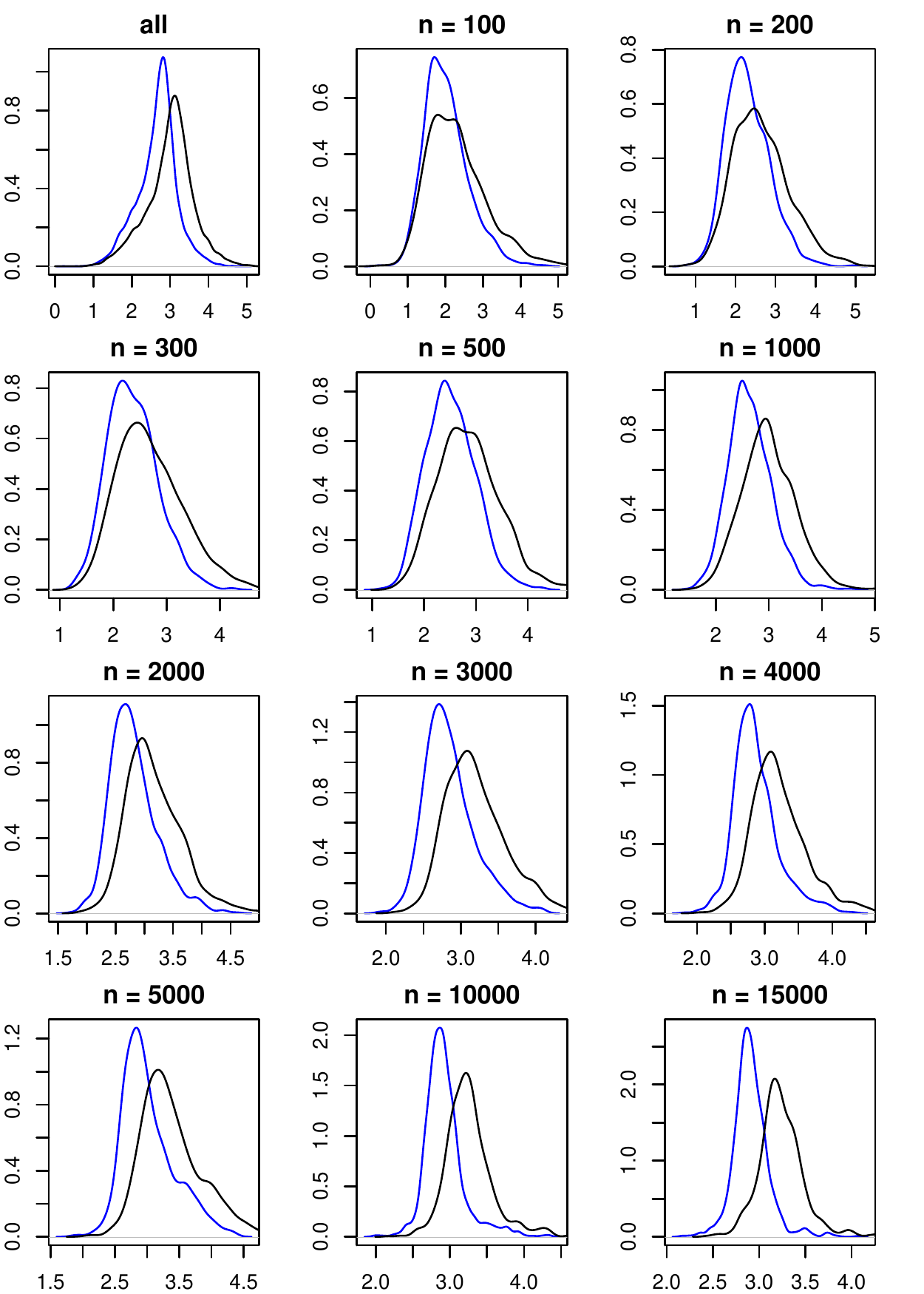}
\end{center}
\caption{\label{ercfig}Density of $L_2/L_g$ (black) and $L_7/L_g$ (blue) for $T = 0.1$.}
\end{figure}

Figure \ref{ercfig} shows the density graph of $L_2/L_g$ and $L_7/L_g$ for every $n$.
Since $L_g$ is negative, 
larger values of these ratios are worse. We know that an optimal embedder should achieve log-likelihood at least as
good as $L_g$. We find that our algorithm yields a significant improvement.

\begin{table}[h]
\begin{center}
\begin{tabular}{lrrrrr}

\toprule

\multirow{2}{*}{$n$} &
\multirow{2}{*}{improved} &
\multicolumn{2}{c}{rel. to 0} & \multicolumn{2}{c}{groundtruth} \\ 

& & avg & median & avg & median \\ \midrule 

100 & 84.6 & 12.8 & 11.9 & 36.2 & 22.0  \\ 
200 & 95.2 & 12.5 & 12.1 & 20.6 & 20.4  \\ 
300 & 97.2 & 11.8 & 10.6 & 18.6 & 17.5  \\ 
500 & 99.1 & 11.5 & 11.1 & 18.0 & 17.5  \\ 
1000 & 99.6 & 11.4 & 11.0 & 17.2 & 16.8  \\ 
2000 & 100.0 & 11.0 & 10.7 & 16.2 & 15.9  \\ 
3000 & 100.0 & 10.9 & 10.7 & 15.9 & 15.8  \\ 
4000 & 100.0 & 10.6 & 10.4 & 15.4 & 15.2  \\ 
5000 & 100.0 & 10.5 & 10.3 & 15.0 & 14.7  \\ 
10000 & 100.0 & 10.3 & 10.2 & 14.8 & 14.7  \\ 
15000 & 100.0 & 10.0 & 9.9 & 14.5 & 14.4  \\ 
\bottomrule
\end{tabular}
\end{center}
\caption{Changes in log-likelihood after applying our procedure to BFKL embedders. 
Percent of improvements signifies the percent of the cases the log-likelihood increased (improved). 
Average and median improvements computed conditionally on improvement. 
In the "rel. to 0" columns, we present the values of $100 \cdot (1-L_7/L_2)$; in the "groundtruth" column,
we present the values of $100 \cdot (1-(L_g-L_7)/(L_g-L_2))$. 
\label{erc_eerc}}
\end{table}

According to the data in Table~\ref{erc_eerc},
we notice that our procedure leads to better embeddings than the pure BFKL embedder no matter the size of the graph. Our procedure yields log-likelihoods
that are closer to the log-likelihood of the groundtruth. The improvement towards
groundtruth is not stable; with the increase of the graph the BFKL embeddings converge, so our improvements become less prominent (around 15\% for large graphs). However, the improvements are statistically significant (p-values are always 0.00 for paired Wilcoxon test with alternative hypothesis that the values of log-likelihood increased after our procedure).

In real-life cases, hardly do we know the groundtruth; comparison of log-likelihoods for BFKL embedder and our procedure
resembles what would we do with real data (columns rel. to zero in Table~\ref{erc_eerc}). In such a case, we may expect that our
procedure should improve the result by average by 10\%. This result
seems stable no matter the size of the graph.

\subsection{Greedy routing}
One potential application of hyperbolic embedding is greedy routing \cite{bogu_internet,tobias_alenex}.
A node $v$ obtains a packet to node $w$; if $w$ is not directly connected to $v$, $v$ needs to select
one of its neighbors through which the packet will be forwarded. In the greedy routing approach, we use the
embedding to select the connected node which is the closest to the goal $w$. Greedy routing fails
if, at some point in the chain, none of the connected nodes is closer to $w$ than $v$ itself.
In \cite{bogu_internet} a hyperbolic embedding of the Internet is constructed, yielding 97\% success rate
of greedy routing. This is much better than a similar algorithm based on geographical placement of
nodes (14\%). High success rate is also robust with respect to link removals \cite{bogu_internet}. While the MLE method of finding
embedding yields worse results than embeddings constructed specifically for the purpose of greedy routing \cite{tobias_alenex},
it is interesting to see how good our methods are according to this metric.


\shortonly{
\begin{table*}
\centering
\scalebox{.9}{
\begin{tabular}{cl|rrr|rrr|rrrr}
\toprule
$T$ &       & \multicolumn{3}{c}{deterioration}    & \multicolumn{3}{c}{improvement}    & \multicolumn{4}{c}{improvement} \\ 
tiling & $n$   & \multicolumn{3}{c}{(discretization)} & \multicolumn{3}{c}{(local search)} & \multicolumn{4}{c}{(three steps)} \\ 
&       & \% & med & avg & \% & med & avg & \% & med & avg & p-value \\ \midrule
        &100      &  95.3  & 3.28 & 5.14 & 81.1 & 2.55 & 3.71 &  56.4 & 1.11 & 1.81 & 0.001 \\ 
        &200      &  97.3  & 3.27 & 4.70 & 82.5 & 2.24 & 3.10 &  53.8 & 0.99 & 1.35 & 0.037 \\ 
        &300      &  99.0  & 3.25 & 4.37 & 85.7 & 1.97 & 2.51 &  55.3 & 0.91 & 1.21 & 0.000 \\ 
        &500      &   99.7  & 3.19 & 4.34 & 87.9 & 1.81 & 2.38 &  55.7 & 0.71 & 1.00 & 0.005 \\ 
        &1000     &  100.0 & 3.25 & 3.92 & 90.8 & 1.66 & 1.99 &  54.0 & 0.51 & 0.70 & 0.071 \\ 
T=0.1   &2000     &  100.0 & 3.27 & 3.73 & 94.3 & 1.58 & 1.73 &  51.6 & 0.40 & 0.50 & 0.806 \\ 
$\ghoz$ &3000     &  100.0 & 3.27 & 3.61 & 98.6 & 1.53 & 1.68 &  49.7 & 0.33 & 0.40 & 0.999 \\ 
        &4000     &  100.0 & 3.33 & 3.52 & 98.6 & 1.55 & 1.63 &  49.0 & 0.29 & 0.37 & 0.999 \\ 
        & 5000     &   100.0 & 3.31 & 3.52 & 98.4 & 1.52 & 1.61 &  46.9 & 0.27 & 0.36 & 0.999 \\ 
        &10000    &  100.0 & 3.27 & 3.37 & 99.8 & 1.56 & 1.60 &  42.1 & 0.19 & 0.23 & 1.000 \\ 
        & 15000    &   100.0 & 3.37 & 3.39 & 99.9 & 1.61 & 1.62 &  32.2 & 0.16 & 0.19 & 1.000 \\ \midrule
        &100      &  86.4   & 1.65 & 2.76 & 72.5 & 1.67 & 2.59 & 61.0 & 1.00 & 1.68 & 0.00  \\ 
        &200      &  90.8   & 1.67 & 2.49 & 78.9 & 1.61 & 2.25 & 63.1 & 0.94 & 1.42 & 0.00  \\ 
        &300      &  94.5   & 1.59 & 2.30 & 82.0 & 1.42 & 1.97 & 65.9 & 0.84 & 1.19 & 0.00  \\ 
        & 500      &   98.2   & 1.57 & 2.16 & 86.8 & 1.35 & 1.69 & 67.6 & 0.80 & 1.04 & 0.00 \\ 
        &1000     &  99.4   & 1.65 & 1.96 & 93.9 & 1.23 & 1.42 & 70.4 & 0.59 & 0.73 & 0.00  \\ 
T=0.1   &2000     &  99.8   & 1.61 & 1.82 & 94.3 & 1.16 & 1.27 & 72.6 & 0.48 & 0.60 & 0.00  \\ 
$\ghoo$ &3000     &  99.8   & 1.64 & 1.75 & 96.4 & 1.15 & 1.21 & 74.2 & 0.42 & 0.49 & 0.00  \\ 
        &4000     &  100.0  & 1.65 & 1.76 & 98.1 & 1.11 & 1.19 & 74.0 & 0.40 & 0.45 & 0.00  \\ 
        & 5000     &  100.0  & 1.63 & 1.75 & 98.3 & 1.15 & 1.20 & 78.2 & 0.40 & 0.46 & 0.00  \\ 
        &10000    &  100.0  & 1.61 & 1.67 & 99.6 & 1.17 & 1.19 & 82.5 & 0.36 & 0.38 & 0.00  \\  
        & 15000    &  100.0  & 1.64 & 1.66 & 99.7 & 1.19 & 1.20 & 87.0 & 0.32 & 0.33 & 0.00 \\ \midrule
        & 300   & 99.9  & 3.88 & 4.53 & 88.5 & 2.27 & 2.72 & 60.7 & 1.04 & 1.37 & 0.00 \\ 
        & 500   & 100.0 & 3.60 & 4.26 & 88.6 & 1.74 & 2.15 & 59.4 & 0.84 & 1.11 & 0.00 \\ 
        & 1000  & 100.0 & 3.39 & 3.73 & 88.1 & 1.25 & 1.48 & 55.9 & 0.65 & 0.78 & 0.00 \\ 
T=0.7   & 2000  & 100.0 & 3.21 & 3.47 & 87.7 & 0.99 & 1.15 & 52.9 & 0.46 & 0.59 & 0.02 \\ 
$\ghoz$ & 3000  & 100.0 & 2.87 & 3.00 & 84.4 & 0.72 & 0.84 & 46.7 & 0.30 & 0.42 & 0.98 \\ 
        & 4000  & 100.0 & 3.22 & 3.31 & 89.7 & 0.80 & 0.88 & 38.3 & 0.28 & 0.37 & 1.00 \\ 
        & 5000  & 100.0 & 2.98 & 3.02 & 88.6 & 0.63 & 0.69 & 36.1 & 0.19 & 0.24 & 1.00 \\ 
        & 10000 & 100.0 & 2.95 & 2.96 & 90.6 & 0.55 & 0.60 & 29.1 & 0.13 & 0.22 & 1.00 \\ 
        & 15000 & 100.0 & 3.23 & 3.23 & 95.5  & 0.73 & 0.75 & 24.4 & 0.12 & 0.16 & 1.00 \\ \midrule
        & 300   & 98.3  & 1.86 & 2.25 & 81.1 & 1.56 & 1.90 & 63.5 & 1.00 & 1.32 & 0.00 \\ 
        & 500   & 99.5  & 1.81 & 2.14 & 85.9 & 1.30 & 1.57 & 68.1 & 0.94 & 1.12 & 0.00 \\ 
        & 1000  & 99.9  & 1.65 & 1.86 & 84.0 & 0.97 & 1.11 & 67.5 & 0.62 & 0.75 & 0.00 \\ 
T=0.7   & 2000  & 100.0 & 1.61 & 1.71 & 87.1 & 0.75 & 0.83 & 67.3 & 0.44 & 0.55 & 0.00 \\ 
$\ghoo$ & 3000  & 100.0 & 1.40 & 1.48 & 83.4 & 0.53 & 0.61 & 64.1 & 0.32 & 0.21 & 0.00 \\ 
        & 4000  & 100.0 & 1.61 & 1.64 & 90.5 & 0.56 & 0.62 & 61.7 & 0.28 & 0.36 & 0.00 \\ 
        & 5000  & 100.0 & 1.48 & 1.50 & 88.8 & 0.45 & 0.49 & 59.5 & 0.25 & 0.28 & 0.00 \\ 
        & 10000 & 100.0 & 1.47 & 1.48 & 94.1 & 0.41 & 0.46 & 61.0 & 0.18 & 0.22 & 0.00 \\ 
        & 15000 & 100.0 & 1.61 & 1.61 & 98.4 & 0.53 & 0.53 & 65.1 & 0.17 & 0.19 & 0.00 \\ \bottomrule
\end{tabular}}
\caption{Changes in the success rate of greedy routing. Average and median computed conditionally on change
(improvement on the condition of the improvement or deterioration on the condition of the deterioration). P-values for Wilcoxon paired tests.
\label{routing}\label{prouting}\label{routing_70}\label{grouting_70}}
\end{table*}
}

Table~\ref{routing} summarizes the changes in success rates after our procedure. We discuss the conditional changes (improvement on the condition of the improvement
or deterioration on the condition of the deterioration), because we find the absolute changes possibly misleading for small networks.
Percent of changes proxies us the probability of the effect. If the effect occurs, we know what to expect without the bias of the counter-effect. 
The success rates of the original embedding are around 93\% on average.
Discretization usually worsens the success probability. This appears to be caused by the fact that two neighbors
of $v$ can be in the same distance to $w$ (because of discretization), while originally the more useful node is closer.
We notice that if the deterioration due to discretization occurs,
the average percentage deterioration decreases with the increasing size of the graph. This is expected, since the distances are larger in larger graphs.
Meanwhile,
the median deterioration (if the deterioration occurs) is stable at around 3\%. This means that with the increasing size of the network we
face serious deteriorations after discretization less often. The effect is statistically significant (p-values of Wilcoxon paired tests with the alternative hypotheses that the success rate is lower after the discretization are always 0.00).
The whole procedure improves the results in more than 45\% of cases, however the change for bigger graphs is not statistically significant (p-values for paired Wilcoxon tests with the
alternative hypotheses that the success rates increased after our procedure are greater than 10\%). For large graphs (over 4000 vertices) in about half of the cases the success rate after the procedure is not worse than the original one.

To reduce the negative effect of discretization, we also perform the same experiment using the $\ghoo$ grid. 
The results are shown in Table \ref{prouting}.
Contrary to the $\ghoz$ tessellation, usage of the finer tessellation for routing (when all three steps are performed)
improves the success rate, and the effect is statistically significant. 
The general directions of the
effects resemble the case of coarser grid.
Discretization leads to a statistically significant decrease in the success rate; the bigger
graphs, the less frequent a noticeable deterioration.

\subsection{Changing the temperature}

\begin{figure}
\begin{center}
\shortonly{\includegraphics[width=\linewidth]{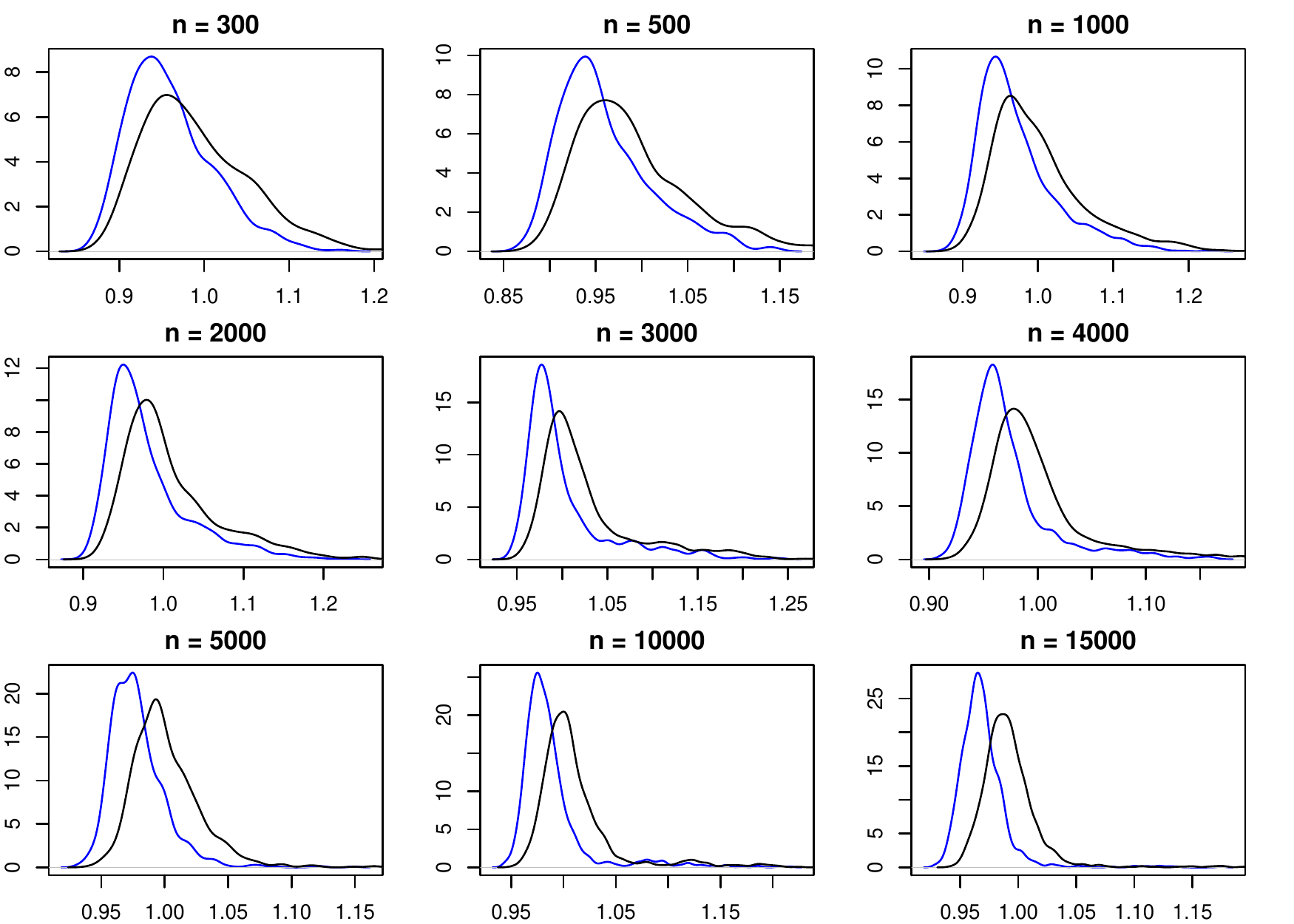}}
\end{center}
\caption{Density of $L_2/L_g$ (black) and $L_7/L_g$ (blue) for $T=0.7$.\label{gtseven}}
\end{figure}

Real-world graphs are considered to have larger temperature $T$ than 0.1.
We have also experimented with changing the temperature $T$ to 0.7.
This value of $T$ has been used for mapping the Internet \cite{bogu_internet}.

For $T=0.7$, the embedder and local search actually tend to achieve a better
result than the groundtruth. For all sizes of graphs, we observe about $3\%$ 
improvement in the absolute value of log-likelihood on average ($1-L_2/L_7$)
(Figure \ref{gtseven}). This effect is statistically significant
(p-values for paired Wilcoxon tests with alternative hypotheses that the values of log-likelihood increased after our procedure always 0.00).

One could argue that the temperature plays critical role for the success of the greedy routing. 
With a higher temperature, the links are less predictable, and thus we can expect a lower success rate.
We conducted the simulations to check if the insights change in the case of the value of temperature more typical
to real-world networks. Table \ref{routing_70} 
contains the results of the experiments with respect to the size of the graph.

Our insights driven from simulations for $T = 0.7$ resemble the conclusions for the case of $T= 0.1$. In the case of the coarse tessellation (Table \ref{routing_70})
the possibility of improvement depends on $n$. For small and medium-sized graphs the effect is statistically significant; the larger the graph, the less probable statistically significant improvement (p-values of Wilcoxon paired test greater that any conventional significance levels). With the finer tessellation (Table \ref{grouting_70}), no matter the
size of the graph, the post-procedure improvement is statistically significant%
\longonly{ (p-values for paired Wilcoxon tests with the alternative hypothesis that the success rate increased after our procedure are always 0.00)}.

We also checked if the change of the tessellation significantly improves the success rate. To this end, we perfomed paired Wilcoxon tests with the alternative
hypotheses that the success rate increased after using a finer tessellation. In all the cases,
the p-values were 0.00, so 
the effect of the tessellation is
statistically significant.

\section{Conclusion}
We introduced the discrete version of the HRG model, 
which allows
efficient algorithms while avoiding numerical issues.
We also presented the result of the experimental evaluation of this model. 
We analyzed both the real-world networks and 20,000 artifical ones,
paying special attention to the possible application of the model in greedy routing.
Our experimental evaluation shows that we achieve a good approximation of log-likelihood
in the 
HRG
model and that using local search significantly improves its log-likelihood, even when converting back to HRG. 
This 
is visible both in real-world and in simulated networks.
A similar procedure also slightly improves the success rate of greedy routing when a sufficiently
fine tessellation is used.
The choice of the tessellation seems to be crucial for the success rate for all tested values of the parameters.

\section*{Acknowledgments}

We would like to thank all the referees for their comments which have greatly improved the paper.
This work has been supported by the National Science Centre, Poland, grant DEC-2016/21/N/HS4/02100.
\putbib
\end{bibunit}
 
\appendix
\newpage

\begin{bibunit}

\section{Tessellation distances versus hyperbolic distances}\label{appendix:versus}

\begin{table}[h!]
\begin{center}
\begin{tabular}{rrrr}
$d$ & $n_d$ & $ER$ & $\mbox{Var}(R)$ \\
 0 &        1 &   0.00000000 & 0.00000000 \\
 1 &        7 &   1.09054966 & 0.00000000 \\
 2 &       21 &   2.02973974 & 0.02467308 \\
 3 &       56 &   2.99594181 & 0.03923368 \\
 4 &      147 &   3.96512066 & 0.05370197 \\
 5 &      385 &   4.93471877 & 0.06823850 \\
 6 &     1008 &   5.90437726 & 0.08279668 \\
 7 &     2639 &   6.87404448 & 0.09735968 \\
 8 &     6909 &   7.84371297 & 0.11192362 \\
 9 &    18088 &   8.81338165 & 0.12648772 \\
10 &    47355 &   9.78305035 & 0.14105185 \\
11 &   123977 &  10.75271905 & 0.15561599 \\
12 &   324576 &  11.72238775 & 0.17018013 \\
13 &   849751 &  12.69205646 & 0.18474427 \\
14 &  2224677 &  13.66172516 & 0.19930841 \\
15 &  5824280 &  14.63139386 & 0.21387255 \\
16 & 15248163 &  15.60106257 & 0.22843669
\end{tabular}
\caption{Tessellation distances ($d$) versus expected hyperbolic distances ($ER_d$): expected value and variance. 
\label{distversus}}
\end{center}
\end{table}

Table \ref{distversus} contains the example detailed results for $\ghoz$ (see Section \ref{sec:scale-free}).
The asymptotic values are obtained in the following way: the difference $ER_d-ER_{d-1}$ converges very quickly to 0.9696687, and the difference
$ER_d-0.9696687d$ converges very quickly to 0.0863634.

\section{Proofs and pseudocodes}\label{appendix:proofs}
\def\bbR{\mathbb{R}}
\def\bbN{\mathbb{N}}

\begin{proof}[Generalization of Proposition \ref{candodist}] \rule{0cm}{0cm}
For $D(G)>1$, there is one more type of a canonical path possible, where
$w$ is a parent of $v$, but neither the leftmost nor rightmost one.

The idea of the algorithm is to find the shortest canonical path.
Suppose that $\dist_0(v) = d' + \dist_0(w)$, where $d' \geq 0$.
For each
$i$ starting from 0 we compute the endpoints of the segments $P^{d'+i}(v)$ and
$P^i(w)$. We check whether these segments are in distance at most $\tlimit$ on the ring;
if no, then we can surely tell that we need to check the next $i$; if yes,
we know that the shortest path can be found on one of the levels from $i$ to 
$i+\lfloor \tlimit/2 \rfloor$. We compute the length of all such paths and return the
minimum. 

The pseudocode of our algorithm is given below.
It uses five integer variables $a_i,d_i,d$ and four 
tile variables $l_i$, $r_i$ ($i=1,2$). Variables $a_i$, $d_i$, $l_i$ and $r_i$
are modified only by the function push($i$), which lets us keep the following
invariant: $\dist_0(l_i) = \dist_0(r_i) = d_i$, $l_i = p_L^{a_i}(v_i)$, 
$r_i = p_R^{a_i}(v_i)$. By $v+k$, where $v$ is a tile, we denote the $k$-th right sibling of $v$. 

The lines (\ref{easystart}-\ref{easyend}) deal with the special case for $D(G)>1$ mentioned above.

The main loop in lines (\ref{mainloopstart}-\ref{mainloopend}) deals with the other cases. At all times
$d$ is the currently found upper bound on $\dist(v,w)$. It is easy to check that
the specific shortest path given in Proposition \ref{candodist} will be 
found by our algorithm.

Every iteration of every loop increases $a_1$ or $a_2$, and an iteration can occur
only if $a_1+a_2 < \dist(v,w)$. Therefore, the algorithm runs in time $O(\dist(v,w))$.

\begin{enumerate}
\itemsep 0em
\def\i{\item}
\def\iz{\item}
\def\cindent{\hskip 0.5cm}

\iz {\bf function} {\sc Distance}$(v_1, v_2)$:
\i  \cindent {\bf for} $i \in \{1,2\}$:
\i  \cindent \cindent $l_i := v_i$
\i  \cindent \cindent $r_i := v_i$
\i  \cindent \cindent $d_i := \dist_0(v_i)$
\i  \cindent \cindent $a_i := 0$
\i  \cindent {\bf function} push($i$):
\i  \cindent \cindent $a_i := a_i+1$
\i  \cindent \cindent $d_i := d_i-1$
\i  \cindent \cindent $l_i := p_L(l_i)$
\i  \cindent \cindent $r_i := p_R(r_i)$
\i  \cindent {\bf while} $d_1 > d_2:$
\i  \cindent \cindent push(1)
\i  \cindent {\bf while} $d_2 > d_1:$
\i  \cindent \cindent push(2)
\i  \label{easystart} \cindent {\bf for} $i \in \{1,2\}$ {\bf if} $v_i \in [l_i, r_i]:$
\i  \label{easyend} \cindent \cindent \cindent return $a_{3-i}$
\i  \cindent $d := \infty$
\i  \label{mainloopstart} \cindent {\bf while} $a_1 + a_2 < d$:
\i  \cindent \cindent {\bf for} $i \in \{1,2\}$ {\bf for} $k \in \{0,\ldots,D(G)\}$ {\bf if} $l_i = r_{3-i} + k:$
\i  \cindent \cindent \cindent $d := \min(d, a_1 + a_2 + k)$
\i  \cindent \cindent push(1)
\i  \label{mainloopend} \cindent \cindent push(2)
\i  \cindent {\bf return} $d$
\end{enumerate}
\end{proof}

\begin{proof}[Proof of Theorem \ref{tallythm} for the general case]
A segment is {\bf good} if it is of the form $P^d([v,v])$ for some $v \in V$ and
$d \in \bbN$. 
In our algorithm the operation {\sc Add}($w$) will update the information in the good segments
of the form $P^d([w,w])$, and the operation {\sc Count}($v$) will follow the algorithm from Proposition
\ref{candodist}, but instead of considering the single segment $[w,w]$, it will count all of them,
by using the information stored in the segments close to $P^d([v,v])$.
Our algorithm will optimize by representing all the good segments
coming from tiles $v$ added to our structure.

We call a tile or good segment {\it active} if it has been already generated,
and thus is represented as an object in memory.
For each active tile $v \in V$ we keep two lists $L_L(v), L_R(v)$
of active segments $S$ such
that $v$ is respectively the leftmost and rightmost element of $S$. 
Each active segment $S$ also has a pointer to
$P(S)$, which is also active (and thus, all the ancestors of $S$ are active too), and a
dynamic array of integers $a(S)$. Initially, there are no active tiles or good segments;
when we activate a segment $S$, its $a(S)$ is initially filled with zeros.
The value of $a(S)[i]$ represents the total $f(w)$ for all tiles $w$ which yield 
the segment $S$ after $i$ operations of the algorithm from the proof of Proposition \ref{candodist},
i.e., $a(S)[i] = \sum_{w: P^i[w,w] = S} f(w)$.

The operation {\sc Add}($w$, $k$) works as follows: for each $i = 0, \ldots, \dist_0(v)$, we simply add $k$ to $a(P^i(S))[i]$.
In the pseudocode below, we assume that $P(S)$ returns {\bf null} if $S$ is the root segment.

\begin{enumerate}
\itemsep 0em
\def\i{\item}
\def\iz{\item}
\def\cindent{\hskip 0.5cm}

\iz {\bf function} {\sc Add}$(w, x)$:
\i  \cindent S := $[w,w]$
\i  \cindent i := 0
\i  \cindent {\bf while} $S \neq$ {\bf null}:
\i  \cindent \cindent $a(S)[i] = a(S)[i] + x$
\i  \cindent \cindent $S := P(S)$
\i  \cindent \cindent $i := i + 1$
\end{enumerate}

The operation {\sc Count}($v$) activates $v$ and $S = [v,v]$ together with all its ancestors.
We return the vector $A$ obtained as follows.
We look at $P^i(S)$ for $i = 0, \ldots, \dist_0(v)$, and for each $P^i(S)$,
we look at close good segments $q'$ on the same level, 
lists $L_L(w), L_r(w)$
for all $w$ in distance at most $\tlimit$ from $P_i(S)$.
The intuition here is as
follows: the algorithm from Proposition \ref{candodist}, on reaching $p^{i_1}(v)=S$ and $p^{i_2}(w)=S'$, 
would find out that these two pairs are close enough and return $i_1+i_2+\dist(S,S')$; in our
case, for each $c$ such that $a(S')[c] \neq 0$, we will instead add $a(S')[c]$ to
$A[a_1+\dist(S,S')+c]$. 

We have to be careful that, if we count some vertex $v$ when considering the pair of segments $(S,S')$, we
do not count it again when considering the pair of segments $(P^j(S),P^j(S'))$. This is done in lines
\ref{nocounttwicea}--\ref{nocounttwiceb} in the pseudocode below. By $S^L$ and $S^R$ we respectively
denote the leftmost and rightmost vertex of the segment $S$.

\begin{enumerate}
\itemsep 0em
\def\i{\item}
\def\iz{\item}
\def\cindent{\hskip 0.5cm}

\iz {\bf function} {\sc Count}$(v)$:
\i  \cindent $U = \emptyset$
\i  \cindent {\bf for each active} $S'  \ni v$:
\i  \cindent \cindent insert$(U, (S', 0))$
\i  \cindent $d := 0$
\i  \cindent $S := [v,v]$
\i  \cindent {\bf while} $S \neq {\bf null}:$
\i  \cindent \cindent {\bf for} $i \in {0, \ldots, D(G)}:$
\i  \cindent \cindent \cindent {\bf for each} $S' \in L_R(S^R+i)$:
\i  \cindent \cindent \cindent \cindent insert$(U, (S', d + \dist(S, S')))$
\i  \cindent \cindent \cindent {\bf for each} $S' \in L_L(S^L-i)$:
\i  \cindent \cindent \cindent \cindent insert$(U, (S', d + \dist(S, S')))$
\i  \cindent \cindent $d := d + 1$
\i  \cindent \cindent $S := P(S)$
\i  \cindent $T = []$
\i  \cindent {\bf for each} $(S', d) \in U$:
\i  \cindent \cindent {\bf for each} i: $T[d+i] = T[d+i] + a(S')[i]$
\i  \cindent \cindent $S''$ := $P(S')$ \label{nocounttwicea}
\i  \cindent \cindent {\bf if} $(S'',d') \in U$ for some $d'$:
\i  \cindent \cindent \cindent {\bf for each} i: $T[d'+i] = T[d'+i] - a(S')[i]$
\i  \cindent \cindent \cindent {\bf break} \label{nocounttwiceb}
\i \cindent {\bf return} $T$
\end{enumerate}
\end{proof}

\section{Details of the GitHub dataset}\label{appendix:github}

In GitHub convention, \emph{following} means a registered user
agreed to be sent notifications about other user's activity within the service.
We represent this relationship using the following graph $\mathcal{G}_f$.
There is an edge in $\mathcal{G}_f$ between A and B if and only if A follows B.
Mechanisms behind the creation of this network involve users'
popularity and the similarity, which suggests underlying hyperbolic geometry of $\mathcal{G}_f$. $\mathcal{G}_f$ also shows power-law-like scale behavior \cite{euromed}; we
believe it is a useful benchmark for our analysis. Since the complete download of GitHub
data is impossible, our dataset is combined from two sources: GHTorrent project
\cite{ght} and GitHubArchive project \cite{gha}. The analyzed networks contain 
information about the following relationships in two snapshots: F09 covers relationships that occurred in the service from 2008 to 2009 and F11 covers the same during 2008-2011 period.

\section{Choice of the tessellation}\label{appendix:choosetes}

\begin{table*}
\begin{center}
{
\begin{tabular}{|l|rrr|rr|rrr|}
\hline
grid      & $L_3$   & $L_5$   & $L_7$   & MB    & \#it & $t_m$ [s] & $t_l$ [s] & $t_e$ [s] \\
\hline                
$G_{810}$ & -187738 & -172018 & -172585 & 46    & 37   & 0.180 & 0.027 & 14.17 \\
$G_{710}$ & -182721 & -170074 & -170873 & 40    & 29   & 0.194 & 0.030 & 12.13 \\
$G_{711}$ & -179125 & -167991 & -168445 & 61    & 23   & 0.281 & 0.058 & 17.82 \\
$G_{720}$ & -179977 & -168105 & -168817 & 98    & 71   & 1.025 & 0.094 & 91.87 \\
$G_{721}$ & -178108 & -167407 & -167824 & 146   & 99*  & 1.359 & 0.208 & 282.0 \\
$G_{753}$ & -177254 & -166889 & -167648 & 1050  & 99*  & 4.446 & 3.059 & 4999  \\
$B_{2}  $ & -180354 & -168055 & -168338 & 47    & 15   & 1.278 & 0.037 & 17.17 \\
$B_{1.1}$ & -180112 & -169019 & -168134 & 54    & 11   & 1.513 & 0.041 & 4.362 \\
$B_{1.0}$ & -179554 & -168717 & -168214 & 53    & 59   & 1.555 & 0.042 & 8.830 \\
$B_{0.9}$ & -179500 & -168973 & -168282 & 56    & 45   & 1.607 & 0.042 & 22.56 \\
$B_{0.5}$ & -179742 & -168906 & -168017 & 62    & 7    & 2.158 & 0.046 & 6.182 \\
$\{5,4\}$ & -195952 & -173641 & -175671 & 38    & 20   & 0.159 & 0.024 & 5.700 \\
\hline                                                   
\end{tabular}}
\vskip 1em
$\begin{smallmatrix} \includegraphics[width=.20\textwidth]{tes/tiling-hep.pdf} \\ G_{710} \end{smallmatrix}$
$\begin{smallmatrix} \includegraphics[width=.20\textwidth]{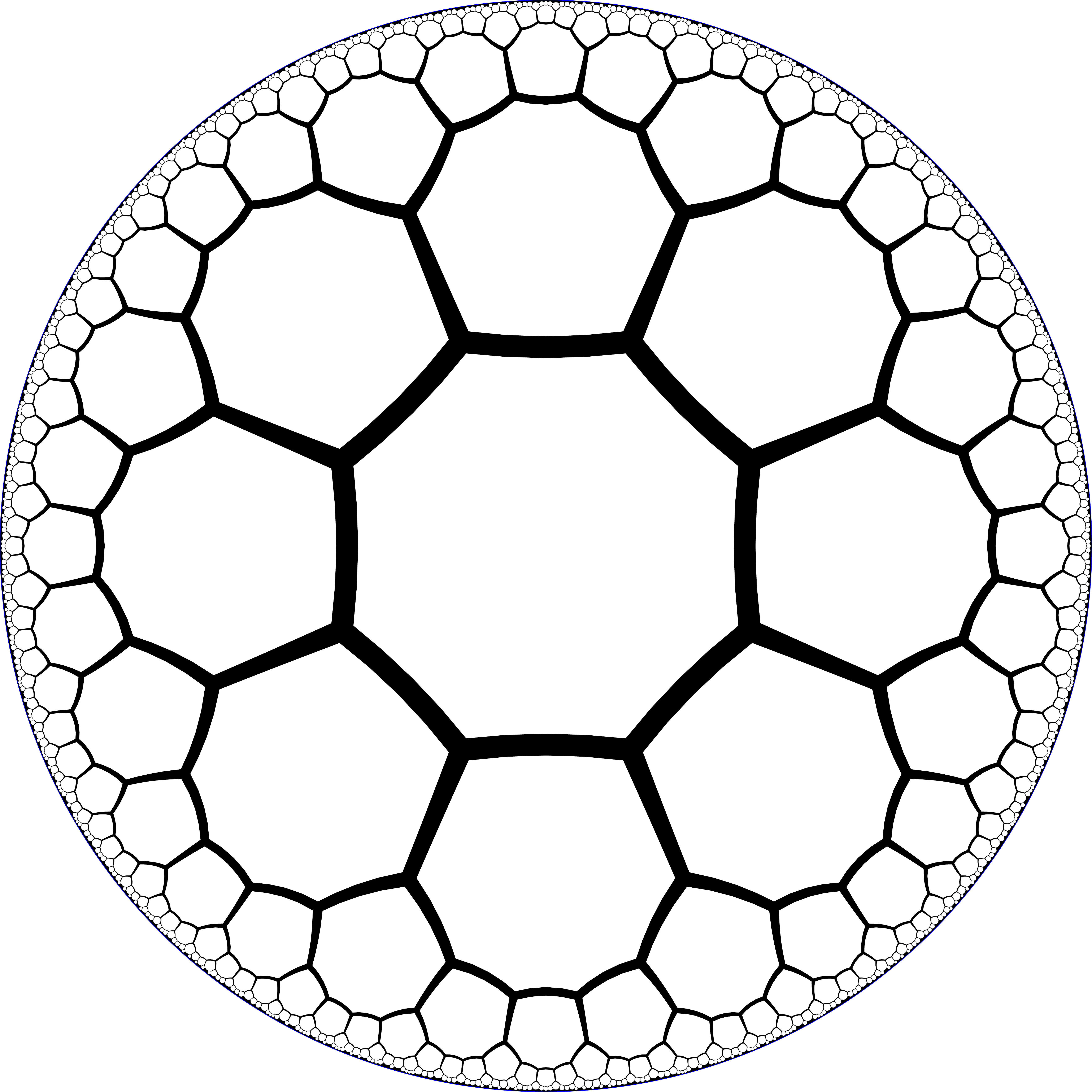} \\ G_{810} \end{smallmatrix}$
$\begin{smallmatrix} \includegraphics[width=.20\textwidth]{tes/tiling-711.pdf} \\ G_{711} \end{smallmatrix}$
$\begin{smallmatrix} \includegraphics[width=.20\textwidth]{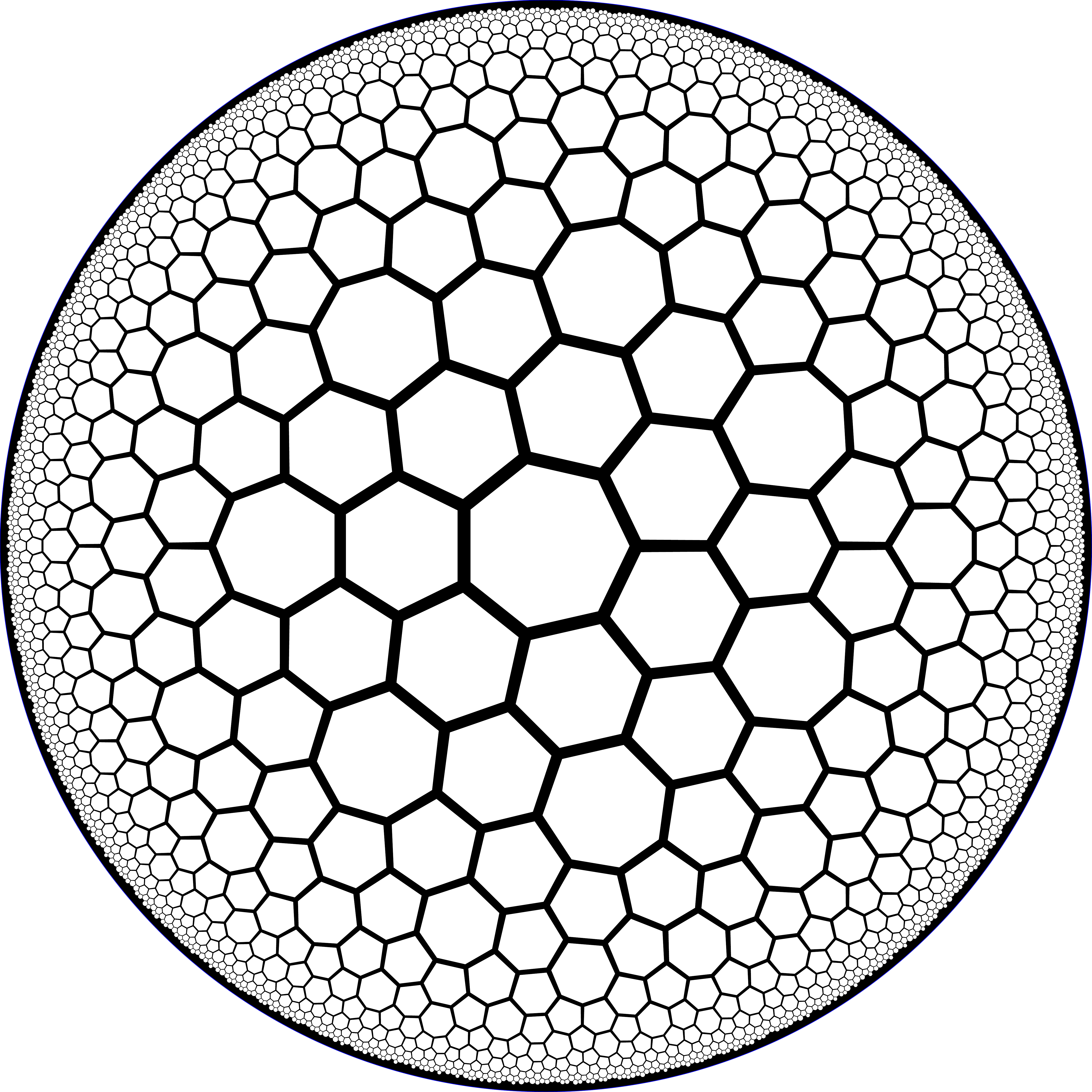} \\ G_{720} \end{smallmatrix}$
$\begin{smallmatrix} \includegraphics[width=.20\textwidth]{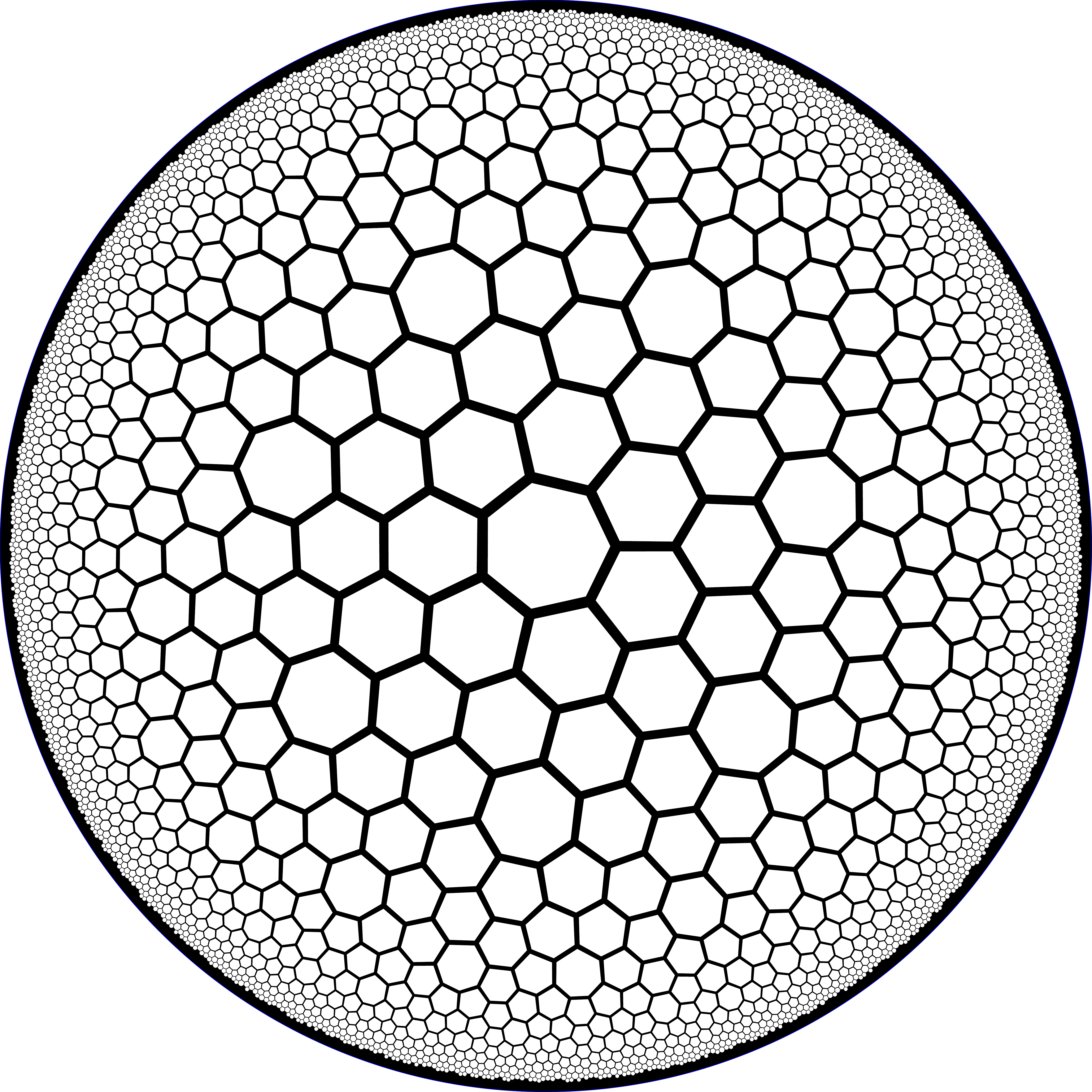} \\ G_{721} \end{smallmatrix}$
$\begin{smallmatrix} \includegraphics[width=.20\textwidth]{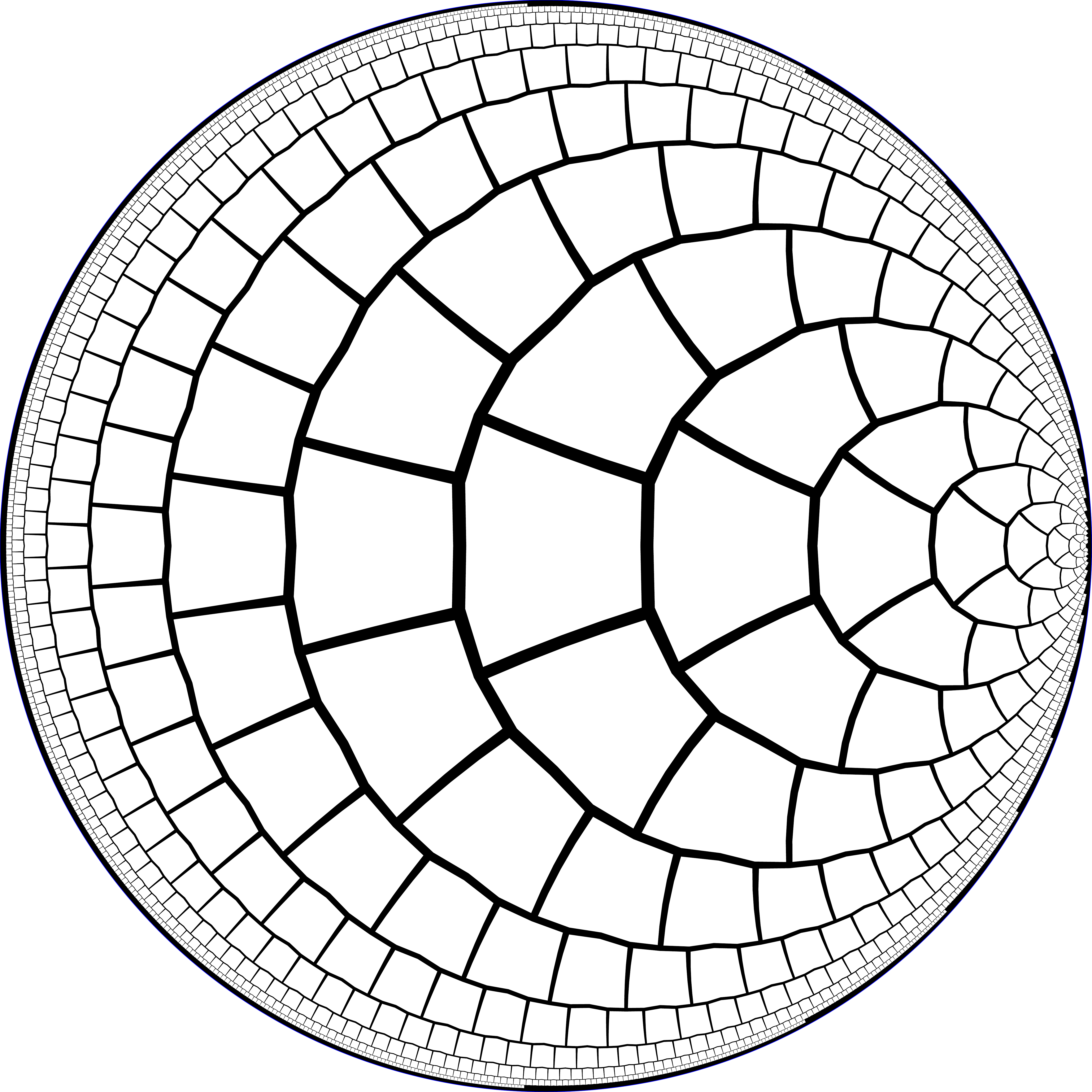} \\ B_{1.0} \end{smallmatrix}$
$\begin{smallmatrix} \includegraphics[width=.20\textwidth]{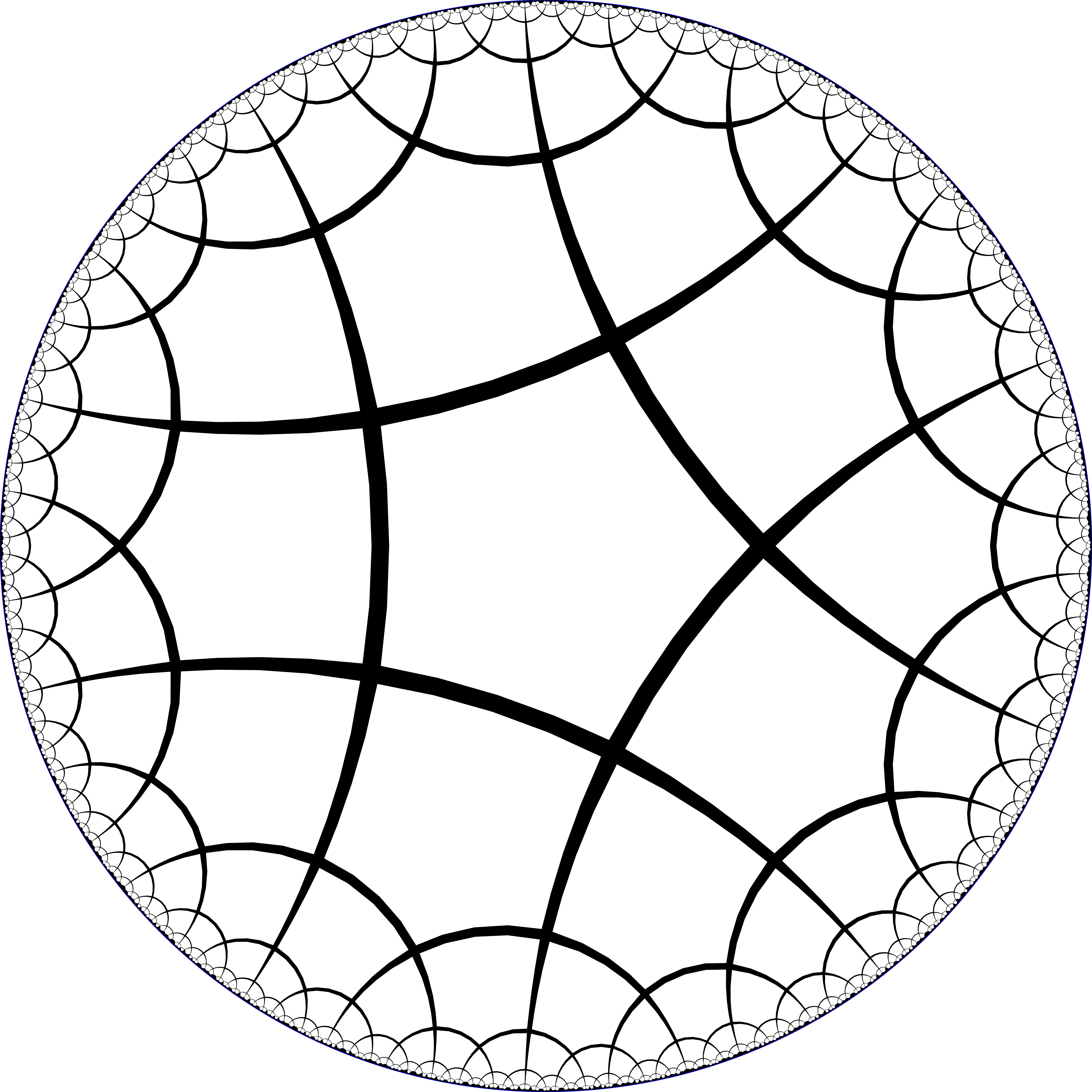} \\ \{5,4\} \end{smallmatrix}$
\end{center}
\caption{Experimental results on the Facebook network ($L_2=-176131$).\label{tab:facebook}}
\end{table*}

Table \ref{tab:facebook} presents the experimental results of running our algorithm on the Facebook social circle network
graph for various tessellations. We can obtain a coarser grid than $\ghoz$ by using octagons instead of heptagons ($G_{810}$),
and a finer grid by using the Goldberg-Coxeter construction, which adds extra hexagons to the order-3 heptagonal tessellation ($G_{7ab}$). We can also
use a tessellation $B_x$ based on the binary tiling \cite{boroczky} (where $x$ is the width of the tile), or $\{5,4\}$, where four pentagons
meet in a vertex. Log-likelihoods are named as in Section \ref{sec:experiments}:
\\
$L_2$, $L_7$ -- continuous best log-likelihoods
\\
$L_3$, $L_5$ -- best discrete log-likelihoods, logistic function
\\
$L_4$, $L_6$ -- best discrete log-likelihood, arbitrary function of distance (used for local search)
\\
The column \#it presents the number of iterations of local search; * denotes that we have stopped the process after
this number of iterations, while no * denotes that the local search could not improve the log-likelihood any further.
MB is the amount of memory in megabytes, and time is in seconds.
\\
As we can see, finer tessellations give better log-likelihoods, but a too dense grid dramatically decreases the performance without
giving significant benefits. Tessellation $B_x$ does not yield significantly better results, despite its circles' greater similarity to continuous ones. The results of $\{5,4\}$ are relatively bad; it approximates distances worse than three-valent
tessellations.



\putbib
\end{bibunit}

\end{document}